\documentclass[pra,floatfix,amsmath,twocolumn]{revtex4}
\usepackage{amssymb}
\usepackage{graphicx}
\usepackage{graphics}
\usepackage{amsmath}
\usepackage{amsthm}

\newcommand{\bea}{\begin{eqnarray}}
\newcommand{\eea}{\end{eqnarray}}
\def\bi{\begin{itemize}}
\def\ei{\end{itemize}}
\def\bc{\begin{center}}
\def\ec{\end{center}}

\def\C{\hbox{$\mit I$\kern-.7em$\mit C$}}
\def\R{\hbox{$\mit I$\kern-.6em$\mit R$}}

\def\ket#1{|#1\rangle}
\newcommand{\one}{\mbox{$1 \hspace{-1.0mm}  {\bf l}$}}
\def\tr{\mathrm{tr}}
\def\ket#1{\left| #1\right>}
\def\bra#1{\left< #1\right|}

\newcommand{\proj}[1]{\ket{#1}\bra{#1}}

\newtheorem{theorem}{Theorem}

\newtheorem{lemma}[theorem]{Lemma}

\begin{document}

\title{Multipartite Entanglement and Global Information}

\author{C. Kruszynska and B. Kraus}

\affiliation{Institute for Theoretical Physics, University of
Innsbruck, Austria}

\begin{abstract}
We investigate the entanglement properties of pure quantum states
describing $n$ qubits. We characterize all multipartite states
which can be maximally entangled to local auxiliary systems using
controlled operations. A state has this property iff one can
construct out of it an orthonormal basis by applying independent
local unitary operations. This implies that those states can be
used to encode locally the maximum amount of $n$ bits. Examples of
these states are the so--called stabilizer states, which are used
for quantum error correction and one--way quantum computing. We
give a simple characterization of these states and construct a
complete set of commuting unitary observables which characterize
the state uniquely. Furthermore we show how these states can be
prepared and discuss their applications.
\end{abstract}
\maketitle

One of the challenges in quantum information theory is to get a
better understanding of multipartite entanglement. Since bipartite
entanglement measures are not sufficient to classify multipartite
entangled states, several other measures of entanglement, like the
tangle \cite{CoKu00} or the localizable entanglement \cite{VePoCi04}
and generalizations of it have been introduced to study the "true"
multipartite entanglement \cite{MiVe04}. Furthermore, different classes of
entangled states have been identified \cite{DuViCi00,VeDeMoVe02} and a normal
form of multipartite states has been presented \cite{Ves}. Several
important applications of multipartite entangled states, like
quantum error correction, quantum computing, but also applications
within condensed matter physics have been found (for recent
reviews see \cite{HoHo07,AmFaOsVe08,WoCi} and reference therein). Despite
all these results, the properties of multipartite entangled states
are far from being completely understood.

Here we use a different approach to gain a new insight into the
entanglement properties of multipartite states. The idea is to
determine how well the qubits can be locally entangled to auxiliary
systems. Before we discuss the operational meaning of this let us
precisely state the situation we investigate here. We consider an
$n$ qubit quantum state $\ket{\Psi}$. Each party uses an auxiliary
qubit to entangle it to its system qubit in such a way that the
global state is a maximally entangled state between the system and
the auxiliary qubits \footnote{We allow only one qubit per site,
because, if we would consider a $4$--level system per side, each
party could implement a completely depolarizing map leading to a
maximally entangled state between the auxiliary systems and the system
qubits.}. The operations which are used by the parties are
so--called controlled operations, which we denote by $C_l$, with $
C_l=\sum_i U_l^i \otimes \ket{i}_{l_a}\bra{i},$ where $U_l^i$ are
unitary operations acting on system $l$ and $\ket{i}_{l_a}\bra{i}$
is acting on the auxiliary system attached to $l$. If there exist
local control gates $C_l$ such that the state $C_1 \otimes C_2
\otimes \ldots \otimes C_n \ket{\Psi}\ket{+}^{\otimes n}$, with
$\ket{+}=1/\sqrt{2}(\ket{0}+\ket{1})$ is maximally entangled
between the system and the auxiliary systems, we call the state
$\ket{\Psi}$ locally maximally entanglable, LME.

The motivations for investigating this class of states are the
following: (1) The characterization of LME states (LMEs) gives a
classification of multipartite entangled states according to their
entanglement properties. To be more precise, let us assume that
$\ket{\Psi}$ is LME. After successfully attaching the auxiliary
qubits locally, all the quantum information contained in
$\ket{\Psi}$ is washed out, since the reduced state describing the
system qubits is maximally mixed. Thus, by local operations it is
possible to wash out all the global correlations of the state. Even
though the local correlations can always be washed out in this way
\footnote{ That is for any single qubit state $\rho$ there exists a
unitary operator $U$, such that $\rho+U\rho U^\dagger =\one$.},
there exist states, e.g. the $\ket{W}$ state \cite{DuViCi00} for
which it is not possible to wash out the global correlations in
this way. Therefore, these states are fundamentally different from
LMEs. (2) We will show that all LMEs can be used for maximal
(local) encoding of classical bits. Let $\ket{\Psi}$ be an
$n$--qubit LME state. Then, each party can locally encode a single
bit value by applying certain unitary operations to the qubit at
his disposal. We will show that the states obtained in this way are
all orthogonal. Thus, they form a maximal set ($2^n$) of globally
perfectly distinguishable states. However, no party can gain
locally any information about the bits owned by the other parties.
(3) LMEs can be used for gate teleportation \cite{Ni01}, i.e.
certain non--local operations can be implemented on an arbitrary
state using LMEs. (4) Many applications of multipartite entangled
states, like quantum error correction \cite{Gothesis97} or the
one--way quantum computer \cite{RaBr01} employ so--called
stabilizer states \cite{Gothesis97} which are LME. Also the
purification schemes studied in \cite{Ben,DuAsBr03,KrMiBrDu06}
purify to LMEs. Therefore, looking at multipartite entanglement
from this different point of view might allow us to generalize
these applications and to find new ones.

The outline of the paper is the following. First we introduce a
standard form of multipartite entangled states. Then we show that a
state $\ket{\Psi}$ is LME iff there exist local unitary operators,
$U_l^{(i_l)}$ such that the set $\{U_1^{(i_1)}\otimes \cdots
\otimes U_n^{(i_n)} \ket{\Psi}\}$ forms an orthonormal basis
(ON--basis), for $i_l\in \{0,1\}$. That is, a state is LME iff it
can be extended to an ON--basis by independent local unitary
operations. Using these results we derive a simple characterization
of all LMEs. In fact, we will show that a state is LME iff it is
local unitary equivalent (LU--equivalent) to a state
$\ket{\Psi}=\sqrt{\frac{1}{2^{n}}}\sum_{i_1,\ldots, i_n=0}^1 e^{i
\alpha_{i_1,\ldots, i_n}}\ket{i_1,\ldots, i_n},$ where
$\alpha_{i_1,\ldots, i_n}\in \R$. The entanglement contained in
this state is completely determined by the classical phases
$\alpha({\bf i})\equiv \alpha_{i_1,\ldots, i_n}$ and their
correlations. We show that all LMEs can be prepared using
generalized phase gates, where the number of qubits on which the
phase gates are acting on depends on the correlations of the phases
$\alpha({\bf i})$. Next, we consider the unitary operators which
correspond, via the Jamio\l kowski isomorphism
\cite{CiDu01,DuCi01}, to the LMEs and show how non--local
operations can be implemented with LMEs. Furthermore, for any LME
state we construct a complete set of commuting unitary observables
(the generalized stabilizer) which uniquely defines the state. This
cannot only be used to construct Hamiltonians for which
$\ket{\Psi}$ is the unique ground state \cite{WoCi}, but also to
design dissipative processes for which $\ket{\Psi}$ is the unique
stationary state \cite{VeWoCi08,KrBu08}. We show that, for
instance, the $3$ qubit $W$ state is not LME, implying that this
state is fundamentally different than, for instance, a GHZ--state.
Furthermore, we show that it is possible to entangle two qubits
locally such that the third party is unable to entangle his system
locally, even if we allow for an arbitrary two--qubit gate. Thus
here, it is impossible to wash out the global correlations using
local unitary operations.

Let us start by introducing our notation. By $X,Y,Z$ we denote the
Pauli operators. The subscript of an operator will always denote the
system it is acting on, or the system it is describing. For
instance $\rho_i$ is the single qubit reduced state of system $i$
of a state $\ket{\Psi}$, i.e. $\rho_i=\tr_{\text{all but
}i}(\proj{\Psi})$ and $\langle W_i\rangle=\tr(\proj{\Psi}W_i)$
denotes the expectation value of the operator $W_i$ acting on
system $i$. $W^i$ denotes the $i$--th power of the operators $W$
with $W^0\equiv \one$ for any operator, $W$. Since we will often
consider local operators and product states we will denote by ${\bf
i}$ the classical bit--string $(i_1,\ldots, i_n)$ with
$i_k\in\{0,1\}$ $\forall k\in \{1,\ldots, n\}$, e.g. $\ket{{\bf
0}}=\ket{0,\ldots, 0}$. We say that a state, $\ket{\Psi}$ is
LU--equivalent to $\ket{\Phi}$ ($\ket{\Psi}\simeq_{LU} \ket{\Phi}$)
if there exist local unitary operators, $U_1,\ldots, U_n$, such
that $\ket{\Psi}=U_1\otimes \cdots \otimes U_n \ket{\Phi}$.

In order to investigate the LMEs we introduce the trace
decomposition of multipartite states. Let $\ket{\Psi}$ be an $n$
qubit state with single qubit reduced states $\{\rho_i\}$. We write
each single qubit reduced state $\rho_i$ in its spectral
decomposition, $\rho_i=U_i^\dagger D_i U_i$, with
$D_i=\mbox{diag}(\lambda^i_1,\lambda^i_2)$, where $\lambda^i_k$ are
the Schmidt coefficients of the bipartite splitting qubit $i$ and
the rest. We call any such decomposition, $U_1 \otimes \cdots
\otimes U_n\ket{\Psi}$, trace decomposition of $\ket{\Psi}$. The
trace--decomposition has the property that the reduced states are
all diagonal in the computational basis, i.e. $\langle
X_i\rangle=\langle Y_i\rangle=0$. In this paper we will only make
use of the trace decomposition. However, it should be noted that
this decomposition can be used to define a unique standard form of
multipartite states \cite{KrKr08b}. For $D_i\not \propto \one$
$\forall i$ the trace decomposition can be easily made unique, by
requiring that $\lambda^i_1\geq \lambda^i_2$, and imposing certain
conditions on the phases of the coefficients of the states in the
computational basis. If $\rho_i=\frac{1}{2}\one$, for some system
$i$, the standard form can be defined as $\lim_{\epsilon\rightarrow
0}\ket{\Psi(\epsilon)}$, where $\ket{\Psi(\epsilon)}$ denotes the
unique standard form of
$\sqrt{1-\epsilon}\ket{\Psi}+\sqrt{\epsilon} \ket{0}^n$
\footnote{Note that there always exists an $\epsilon_0>0$ such that
for all $0<\epsilon<\epsilon_0$ non of the reduced states is
proportional to the identity. Since $\ket{\Psi(\epsilon)}$ is a
continuous function of $\epsilon$ in this region, the limit
exists.}. Any state can be transformed by local unitary operations
into its standard form \cite{KrKr08b}. Thus, it is easy to
verify that if the standard forms of two states are equivalent,
then the states are LU--equivalent. Note that this standard form
coincides for the simplest case of two qubits with the Schmidt
decomposition \cite{NiCh00} and can be generalized to $d$--level
systems.

Let us now characterize the LMEs. We show that a state is LME iff
it is extendable by independent local unitary operations to an
ON--basis.

\begin{lemma} An $n$--qubit state $\ket{\Psi}$ is LME iff there exists for each
party $l$ a unitary operation $U_l$ such that the set \bea
\{U_1^{i_1}\otimes \ldots, \otimes U_n^{i_n}
\ket{\Psi}\}_{i_l=0,1},\eea forms an ON--basis.
\end{lemma}

\begin{proof} Only if: If $\ket{\Psi}$ is LME then there
exist control operations $C_l=\sum_i V_l^{(i)} \otimes
\ket{i}_{l_a}\bra{i}$ such that $\ket{\Phi}=C_1 \otimes C_2 \otimes
\ldots \otimes C_n \ket{\Psi}\ket{+}^{\otimes n}$ is maximally
entangled in the splitting system versus auxiliary systems.
Applying $(V_l^{(0)}) ^\dagger$ to each system $l$ does not change
the entanglement properties and therefore $\rho_{1,\ldots,
n}=\frac{1}{2^n} {\cal E}_1\circ \ldots \circ{\cal E}_n
(\proj{\Phi})=\frac{1}{2^n}\one, $ where ${\cal E}_l(\rho)=\rho+
U_l\rho U_l^\dagger$, with $U_l=(V_l^{(0)}) ^\dagger V_l^{(1)}$. Since ${\cal E}_1\circ \ldots \circ{\cal
E}_n (\proj{\Phi})=\sum_{{\bf i}} \proj{\Psi_{\bf i}}$, with
$\ket{\Psi_{\bf i}}= U_1^{i_1}\otimes \ldots \otimes U_n^{i_n}
\ket{\Psi}$ is a sum of $2^n$ projectors, this can only be
fulfilled if $\{\ket{\Psi_{\bf i}}\}$ is an ON--basis
\footnote{This can only be fulfilled by orthogonal states since the
range of $\sum_{i=1}^{dim({\cal H})} \proj{\Psi_i}$ spans the whole
Hilbert space, ${\cal H}$, iff $\{\ket{\Psi_i}\}$ is linearly
independent. Then, $(\sum_{i=1}^{dim({\cal H})} \proj{\Psi_i})
\ket{\Psi_j}=\ket{\Psi_j}$ iff all states are orthogonal.}. To see
the inverse, one only has to define $C_l= \one \otimes \proj{0}+
U_l\otimes \proj{1}$.
\end{proof}

Note that the proof implies that if $\ket{\Psi}$ is LME then there
exist local unitary operations, $U_i$ such that $\rho_i+ U_i \rho_i
U_i^\dagger=\one$. That is, the local operations which wash out the
global correlations must also wash out the local correlations. We
are going to show now that these unitary operations are of a
special form. Note that $\{U_1^{i_1}\otimes \ldots \otimes
U_n^{i_n} \ket{\Psi}\}_{i_l=0,1}$ is an ON-basis iff $\{(V_1U_1
V_1^\dagger)^{i_1}\otimes \ldots \otimes (V_n U_n
V_n^\dagger)^{i_n} V_1\otimes \cdots V_n\ket{\Psi}\}_{i_l=0,1}$ is
an ON--basis, implying that a state is LME iff any LU--equivalent
state is LME. Therefore, we can restrict ourselves to some trace
decompositions of the state $\ket{\Psi}$, which we denote by
$\ket{\Psi_t}$. For $\rho_i \not\propto \one $ and $\rho_i$
diagonal the necessary condition, $\rho_i+ U_i \rho_i
U_i^\dagger=\one$, can only be fulfilled by $U_i=e^{i\alpha_i/2
Z_i}X_i e^{-i\alpha_i/2Z_i}$. For $\rho_i \propto \one$, we also
find that $U_i=V_i X_i V_i^\dagger$ for some unitary $V_i$ (up to a
global phase) since $\tr(U_i)=0$ follows for the fact that
$U_i\ket{\Psi}$ must be orthogonal to $\ket{\Psi}$ (Lemma $1$).
Thus, we only have to consider $X$ operations which implies that a
state $\ket{\Psi}$ is LME iff $\ket{\Psi}\simeq_{LU} \ket{\Phi}$,
where $\{X^{i_1}\otimes \ldots \otimes X^{i_n} \ket{\Phi}\},$ is an
ON-basis, i.e $\bra{\Phi} X^{i_1}\otimes \ldots \otimes
X^{i_n}\ket{\Phi}=0\;\forall\;{\bf i}\neq {\bf 0}.$ Using all that
it is now easy to show the following theorem.

\begin{theorem} A state $\ket{\Phi}$ is LME iff $\ket{\Phi}$ is
LU--equivalent to a state $\ket{\Psi}$ with \bea \label{LME}
\ket{\Psi}=\sqrt{\frac{1}{2^{n}}}\sum_{{\bf i}}^1 e^{i \alpha_{{\bf
i}}}\ket{{\bf i}}\equiv U^\Psi_{ph}\ket{+}^{\otimes n},\eea where
$\alpha_{\bf{i}}\in \R$ and $U^\Psi_{ph}$ denotes the
diagonal unitary operator with the entries $e^{i\alpha_{{\bf i}}}$
\footnote{Note that these states can be easily transformed to a
trace decomposition by applying the local unitary operations $HU_i$
with $H$ the Hadamard transformation and $U_i$ such that
$U_i\ket{0}=e^{i x_i}\ket{0}$ and $U_i\ket{1}=\ket{1}$, where
$\cot(x_i)= \frac{\langle X_i\rangle}{\langle Y_i\rangle}$.}
.\end{theorem}

\begin{proof} As we have seen before, $\ket{\Phi}$ is LME iff $\ket{\Phi}\simeq_{LU}
\ket{\Psi}$ with $\bra{\Psi}X^{i_1}\otimes \ldots \otimes X^{i_n}
\ket{\Psi}=0$ $\forall {\bf i}\neq {\bf 0}$ or, equivalently,
$\ket{\Phi}\simeq_{LU} \ket{\Psi}$ with $\bra{\Psi}Z^{i_1}\otimes
\ldots \otimes Z^{i_n} \ket{\Psi}=0$ $\forall {\bf i}\neq {\bf 0}$.
We write $\ket{\Psi}$ in the computational basis, $\ket{\Psi}=
\sum_{{\bf i}} \lambda_{{\bf i}}\ket{{\bf i}}$ and use that
$\proj{i_k}=1/2 (\one+(-1)^{i_k}Z)_k$. Then we have $|\lambda_{{\bf
i}}|^2=\langle \proj{{\bf i}}\rangle= 2^{-n} \langle
(\one+(-1)^{i_1}Z)_1\otimes \cdots \otimes
(\one+(-1)^{i_n}Z)_n\rangle.$ Since all expectation values of the
operators where at least one $Z$ operator occurs vanish we have $
|\lambda_{{\bf i}}|^2=2^{-n}.$
 \end{proof}

Thus, a state is LME iff there exists a product basis such that all
the coefficients of the state in this basis are phases. The control
gates to create the maximally entangled state between the system
(described by $\ket{\Psi}$ in Eq. (2)) and auxiliary qubits are the
two--qubit $\pi$--phase gates, $C=\proj{0}\otimes \one +
\proj{1}\otimes Z$. Note that, given an $n$ qubit LME state (Eq.
(2)) one can construct an $n+1$ qubit LME state by entangling an
additional qubit via $C$ to some system $j$. The phases would
change to $\alpha_{i_1,\ldots, i_{n+1}}=\alpha_{i_1,\ldots,
i_{n}}+\pi i_j i_{n+1}$. In this way one can attach arbitrarily
many qubits.

Since there are $2^n$ real parameters many multipartite states have the property of being LME. For instance any two--qubit
state is LME. This can be easily verified using the Schmidt
decomposition (standard form) of the state,
$\ket{\Psi}=\alpha\ket{00}+\sqrt{1-\alpha^2}\ket{11},$ with $\alpha
\in \R, \alpha \geq 0$ and choosing $U_1=X$ and $U_2=Y$. Prominent
examples of LMEs are the stabilizer states (which are LU--equivalent to the graph states) and the weighted graph states
\cite{Gothesis97, Hein}. There the phases $\alpha_{{\bf i}}$ are
quadratic functions of the index ${\bf i}=(i_1,\ldots, i_n)$, i.e.
$\alpha_{{\bf i}}=\pi {\bf i}^T \Gamma {\bf i},$ where the $n\times
n$ matrix $\Gamma$ is the so--called adjacency matrix of the
mathematical graph corresponding to the graph state \cite{Hein}.
Note that any product state is LME, however, it is very simple to
distinguish product states from entangled states using this notion.
If $\ket{\Psi}$ is a product state then the state $C_1\otimes
\ldots\otimes C_n\ket{\Psi}\ket{+}^n$ is maximally entangled
between the system and the auxiliary systems iff each party creates
a maximally entangled state (locally). Thus, considering the
difference between the local entanglement (each qubit with its
auxiliary system) and the global entanglement allows us to
distinguish product states from entangled states. Similar arguments
can be used to distinguish biseparable states from truly
multipartite entangled states \cite{KrKr08b}. In the following we
consider the general LME state $\ket{\Psi}$ given in Eq.
(\ref{LME}) and denote by $\ket{\Psi_{{\bf i}}}\equiv
\ket{\Psi_{i_1,\ldots, i_n}}=Z^{i_1}\otimes \cdots \otimes
Z^{i_n}\ket{\Psi}$ the elements of the ON-basis ($ \ket{\Psi}\equiv
\ket{\Psi_{{\bf 0}}}$) \footnote{For any LU--equivalent state the
same results apply.}.

Let us now discuss some applications of LMEs. An LME state can be
used to encode classical information locally. If $n$ parties share
the LME state $\ket{\Psi}$ (Eq. (\ref{LME})), each party can encode
a single bit value by applying either $\one$ (corresponding to the
bit value $0$), or $Z$ (corresponding to the bit value $1$), to the
qubit at his possession. The $2^n$ states obtained in this way are
globally perfectly distinguishable (since they are all orthogonal
due to Lemma $1$), but locally, no information can be gained. Note
that for instance for the $\ket{W}$ state, which is not LME, as we
shall see below, it is possible to find local unitary operations
$V_i\otimes W_i\otimes U_i$ such that $\{V_i\otimes W_i\otimes U_i
\ket{W}\}$ is an ON--basis \cite{AkBr05}. However, in this case the
unitary operators which generate the ON--basis depend on each other
which prevents us from using the state to encode locally $n$
independent classical bits. Apart from that, LMEs can also be used
to implement certain non--local unitary operations. In order to see
that, we use the Jamio\l kowski isomorphism which is a one--to--one
mapping between quantum states and quantum operations
\cite{CiDu01,DuCi01}. For an LME state $\ket{\Psi}$, the operation
which corresponds to the state $C_1\otimes \ldots\otimes
C_n\ket{\Psi}\ket{+}^{n}$, where $C$ is the two--qubit $\pi$--phase
gate, is unitary and has the form $U_{\Psi}=\sum_{{\bf
i}}\ket{\Psi_{{\bf i}}}\bra{{\bf i}}=U^\Psi_{ph}H^{\otimes n},$
where $H$ is the Hadamard transformation. This implies that having
an LME state $\ket{\Psi}$ at ones disposal, one can implement (up
to local Pauli operators) the unitary operation $U_{\Psi}$ on an
arbitrary state using only local operations
\cite{CiDu01,DuCi01}. Note that $\ket{\Psi}$ can also be employed
to implement certain transformations on a state describing less
than $n$ qubits. For instance, the one--way quantum computer
proposed in \cite{Rau1} uses Cluster states \cite{BrRa01}, for
which $U^\Psi_{ph}$ is a product of two--qubit $\pi$--phase gates
only. Due to the structure of these LMEs, it was possible to show
that any unitary operator (and therefore quantum computing) can be
implemented in this way.

Let us now briefly discuss how LMEs can be generated. We write any
LME state $\ket{\Psi}$ as \bea \ket{\Psi}=U_{1,\ldots, n} \prod
U_{i_{k_1},\ldots, i_{k_{n-1}}}\cdots \prod U_i\ket{+}^n, \eea
where all the operators are phase gates acting on up to $n$ qubits.
For instance, $U_{123}$ maps $\ket{111}_{123}$ to
$e^{i\phi_{123}}\ket{111}_{123}$, with $\phi_{123}\in \R$ and
leaves the rest unchanged. It is straightforward to see that in
this way the $2^n$ phases $\alpha_{{\bf i}}$ can be generated.
Thus, any LME state can be prepared using generalized phase gates,
which could result from a generalized Ising interaction. If
$\alpha({\bf i})$ is a polynomial of degree $k$ (as a function of
${\bf i}=(i_1,\ldots, i_n)$) then the corresponding state can be
prepared using only $k$--body interactions. E.g. graph states or
weighted graph state, where the phases $\alpha_{{\bf i}}$ are
polynomials of degree $2$ can be created using only two--qubit
phase gates. This shows that the correlations in the coefficients
are directly related to a preparation scheme and therefore to the
entanglement contained in the state.

In order to discuss different methods for the preparation of any
LME state $\ket{\Psi}$, we construct a complete set of commuting
unitary and hermitian operators, $\{W_1,\ldots, W_n\}$ which
uniquely define $\ket{\Psi}$ (generalized stabilizers
\cite{Gothesis97}). We define $W_k= U_\Psi Z_k
U_\Psi^\dagger=U^\Psi_{ph} X_k (U^\Psi_{ph})^\dagger$. Then
$W_k\ket{\Phi}=\ket{\Phi}$ $\forall k$ iff $\ket{\Phi}=\ket{\Psi}$
\footnote{Using the Theorem, it is straight forward to show that
$W_k= \sum_{i_1,\ldots, i_n} e^{i \beta_{i_1,\ldots,
i_{k-1},i_{k+1},\ldots, i_n}}\times \proj{i_1}\otimes \ldots
\ket{0}_k\bra{1}\otimes \ldots \proj{i_n}+h.c,$ with
$\beta_{i_1,\ldots, i_{k-1},i_{k+1},\ldots,
i_n}=\alpha_{i_1,\ldots, i_k=0,\ldots i_n}-\alpha_{i_1,\ldots,
i_k=1,\ldots i_n}$. These operators correspond to local operators
iff for all $k$ the phases $e^{i\beta_{i_1,\ldots,
i_{k-1},i_{k+1},\ldots, i_n}}$ can be factorized, i.e. $
e^{i\beta_{i_1,\ldots, i_{k-1},i_{k+1},\ldots, i_n}}=e^{i
f_1(i_1)}\cdots e^{i f_n(i_n)},$ for some functions $f_i$. This is
for instance the case for graph states, where all the operators
$W_k$ are of the form $V_1\otimes \cdots X_k\otimes \cdots \otimes V_n$,
where each $V_i\in\{\one, Z\}$.}. Note that all these unitary
observables have as a common eigenbasis the basis
$\{\ket{\Psi_{{\bf i}}}\}$ and that $W_k^2=\one$. Let ${\cal W}$
denote the group generated by $\{W_1,\ldots, W_n\}$. Then we have,
similarly to the stabilizer states, $ \sum_{W\in {\cal W}}
W=\proj{\Psi}$ \footnote{ This can be easily seen by noting that
$\sum_{W\in {\cal W}}W=U_\Psi \sum_{\bf{i}} Z^{i_1}\otimes \cdots
\otimes Z^{i_n} U_\Psi^\dagger$ and
$\bra{k_1,\ldots,k_n}\sum_{\bf{i}} Z^{i_1}\otimes \cdots \otimes
Z^{i_n}
\ket{k_1,\ldots,k_n}=\Pi_{l\in\{1,\ldots,n\}}((1)+(-1)^{k_l})=\Pi_{l\in\{1,\ldots,n\}}\delta_{k_l,0}$.}.
Depending only on the phases $\alpha_{\bf i}$, which define the
LME, $\ket{\Psi}$, the generalized stabilizer operators can be
quasi--local, i.e. act non trivially on a small set of
(neighboring) qubits \cite{KrKr08b}. In this case, the methods
developed in \cite{VeWoCi08,KrBu08} can be employed to derive a
quasi--local dissipative process for which the unique stationary
state is $\ket{\Psi}$. Apart form that, one can also easily
construct Hamiltonians for which the unique ground--state is
$\ket{\Psi}$, e.g. $H=\one-\sum_{W\in {\cal W}} W$.

As an example of a state which is not LME we consider the three
qubit $W$--state, $
\ket{W}=\frac{1}{\sqrt{3}}\left(\ket{001}+\ket{010}+\ket{100}\right)$.
Due to the fact that $\ket{W}$ is already in its standard
decomposition, the unitary operations we have to consider are of
the form $U_i=e^{i\alpha_i Z_i} X_i e^{-i\alpha_i Z_i}$. Since
$\langle U_i\otimes U_j\rangle\propto \cos(\alpha_i-\alpha_j)$ it
is impossible that all these expectation values vanish for any pair
of unitary operations. As a consequence, it is only possible to
choose $U_1,U_2$ such that the set $\{\ket{W},U_1\otimes
\one\ket{W}, \one \otimes U_2\ket{W}, U_1\otimes U_2\ket{W}\}$ is
orthogonal, for instance with $U_1=X, U_2=Y$. This means, that it
is impossible for the third party to entangle an auxiliary system
such that the $3$ system qubits are maximally entangled to the $3$
auxiliary qubits. One can also show that if two parties maximally
entangled their system qubit with a local auxiliary qubit then the
third party cannot adequately entangle his auxiliary qubit to his
system qubit, even if he would apply a general two-qubit gate
\footnote{Here, one needs to consider instead of the local unitary
operators local POVMs ($\{E_k^{i_k}\}$ with $(E_j^{0})^\dagger
E_j^{0}+(E_j^{1})^\dagger E_j^{1}=\one$). The state would be
locally entanglable iff $\langle (E_1^{i_1})^\dagger
E_1^{j_1}\otimes \ldots \otimes (E_1^{i_n})^\dagger
E_1^{j_n}\rangle=\delta_{i_1,j_1}\ldots \delta_{i_n,j_n}1/2^n$
\cite{KrKr08b}.}. Thus, two of the three parties can lock some
information in the state by entangling their system to auxiliary
systems.


Investigating the entanglement properties of LMEs might lead to an
insight to the entanglement properties of arbitrary many--body
states, since the class we consider is very large ($2^n$ real
parameters). Due to the simplicity of the form of the states and
the underlying physical picture it might be possible to define new
operational entanglement measures. It might also be feasible to
define the MREGS, i.e. the minimal set of reversible entangled
states, for LMEs \cite{PoBe,DuCi03}. Furthermore, this notion can
also be used to study the separability problem \cite{KrKr08b}. We
are planning to generalize the known quantum informational tasks,
which use LMEs, like quantum computing, and quantum communication
tasks \cite{Kempe} employing more general LMEs than stabilizer
states and weighted graph states. This might allow us to find new
applications of multipartite systems and therefore new operational
entanglement measures. Apart from that, considering a restricted
set of LMEs, where for instance only certain three qubit phases
gates are required to generate the states, might allow us to
generalize the well--known Gottesman--Knill theorem
\cite{NiCh00}. Identifying a large enough subset of these
states might also be relevant for the simulation of quantum systems
\cite{VePoCi04s, AnPl06}. Furthermore, the states which are not LME
might be used for locking information and avoiding certain errors.

B. K. would like to thank J. I. Cirac for interesting discussions.
We acknowledge support of the FWF (Elise Richter Program) and the
European Union (OLAQUI, SCALA, QICS).

\end{document}